\documentclass[a4paper,USenglish,cleveref,autoref,thm-restate]{lipics-v2021}
\usepackage{amssymb,amsmath,mathtools,microtype,xspace}
\newif\ifproc
\proctrue

\usepackage[commandnameprefix=ifneeded, authormarkup=none]{changes}

\nolinenumbers

\hideLIPIcs  

\usepackage{dsfont}

\title{On the Stability of Minimum-Weight Perfect Matching on the Line} 

\author{Mark de Berg}{Department of Mathematics and Computer Science, TU Eindhoven, the Netherlands}{M.T.d.Berg@tue.nl}{https://orcid.org/0000-0001-5770-3784}{MdB is supported by the Dutch Research Council (NWO) through Gravitation-grant NETWORKS-024.002.003.}
\author{Ulrike Schmidt-Kraepelin}{Department of Mathematics and Computer Science, TU Eindhoven, the Netherlands}{U.Schmidt.Kraepelin@tue.nl}{https://orcid.org/0000-0002-9213-7746}{}
\author{Andree-Ovidiu Ștef}{Department of Mathematics and Computer Science, TU Eindhoven, the Netherlands}{andreestef34@gmail.com}{}{}
\authorrunning{M.~de Berg and U.~Schmidt-Kraepelin and A.O.~Ștef} 
\Copyright{Mark de Berg and Ulrike Schmidt-Kraepelin and Andree-Ovidiu Ștef} 
\ccsdesc[100]{Theory of computation~Design and analysis of algorithms}
\keywords{Euclidean matching, stable approximation algorithms, dynamic algorithms, online algorithms, bounded recourse} 
\EventEditors{}
\EventNoEds{0}
\EventLongTitle{}
\EventShortTitle{}
\EventAcronym{}
\EventYear{2026}
\EventDate{}
\EventLocation{}
\EventLogo{}
\SeriesVolume{0}
\ArticleNo{0}

\newcommand{\defeq}{:=}
\newcommand{\xor}{\otimes}

\newcommand{\alg}{\text{\textsc{alg}}\xspace}

\newcommand{\Malg}{M_{\mathrm{alg}}}
\newcommand{\Mopt}{M_{\mathrm{opt}}}
\newcommand{\Min}{M_{\mathrm{in}}}
\newcommand{\Mout}{M_{\mathrm{out}}}
\newcommand{\weight}{\mathrm{weight}}

\newcommand{\opt}{\text{\textsc{opt}}}

\newcommand{\Ieven}{I_{\mathrm{even}}}
\newcommand{\Iin}{I_{\mathrm{in}}}
\newcommand{\Iout}{I_{\mathrm{out}}}

\newcommand{\lsbp}[1]{\mathrm{lsbp}(#1)}

\newcommand{\msbp}[1]{\mathrm{msbp}(#1)}

\newcommand{\floor}[1]{\left\lfloor \, #1 \, \right\rfloor}
\newcommand{\ceil}[1]{\left\lceil \, #1 \, \right\rceil}

\newcommand{\Bin}{B^{\mathrm{in}}}

\newcommand{\lp}[1]{\ell(#1)}
\newcommand{\rp}[1]{r(#1)}

\newcommand{\lsplit}[1]{s_1(#1)}
\newcommand{\rsplit}[1]{s_2(#1)}
\newcommand{\Mun}{M_{\mathrm{unw}}}
\newcommand{\Mwr}{M_{\mathrm{wr}}}

\newcommand{\eps}{\varepsilon}

\newcommand{\bcov}{\mu}

\newcommand{\B}{\ensuremath{\mathcal{B}}}
\newcommand{\C}{\ensuremath{\mathcal{C}}}

\newcommand{\cG}{\ensuremath{\mathcal{G}}}

\newcommand{\Reals}{\mathbb{R}}
\newcommand{\Nats}{\mathbb{N}}

\DeclareMathOperator{\symdif}{\Delta}

\renewcommand{\leq}{\leqslant}
\renewcommand{\geq}{\geqslant}

\renewcommand{\ge}{\geqslant}

\colorlet{mdbc}{blue}
\colorlet{uskc}{teal}
\colorlet{asc}{orange}

\newif\ifcomment
\commenttrue
\ifcomment
\newcommand{\mdb}[1]{\textcolor{mdbc}{MdB: #1}}
\newcommand{\usk}[1]{\textcolor{uskc}{Ulrike: #1}}

\newcommand{\as}[1]{\textcolor{asc}{Andree: #1}}
\else
\newcommand{\mdb}[1]{}
\newcommand{\usk}[1]{}

\newcommand{\as}[1]{}
\fi 

\definechangesauthor[name=Andree, color=mdbc]{Mark}
\definechangesauthor[name=Andree, color=uskc]{Ulrike}
\definechangesauthor[name=Andree, color=asc]{Andree}

\begin{document}
\maketitle
\begin{abstract}
Computing a minimum-weight perfect matching for a point set $P$ in Euclidean space 
is a classic geometric optimization problem. We consider the problem in a dynamic
setting, where pairs of points may be added to or removed from the set~$P$. Our focus
is on maintaining an approximately optimal solution without making too many changes to the
solution. More precisely, we are interested in \emph{$k$-stable algorithms}, which
change at most $k$ edges in the matching after each update to the set~$P$.
In other words, we consider an online setting (with insertions and deletions) with \emph{bounded recourse}.

We study trade-offs between the stability of the algorithm and the approximation ratio
of the maintained solution for point sets in~$\Reals^1$. 
First, we present an $O(\sqrt{n})$-stable algorithm that maintains a $2$-approximation,
which we show to be optimal among all algorithms with sublinear stability. 
Second, we prove that any $o(\log n)$-stable algorithm has unbounded approximation ratio.
Our lower bounds hold even in the insertion-only case,
while our algorithm works in the fully dynamic case. 
Moreover, our lower bounds also hold for the bipartite variant of the problem.
\end{abstract}

\section{Introduction}
\label{sec:introduction}
\subparagraph{Background.} 
Computing minimum-weight perfect matchings is a classic problem in combinatorial optimization. 
In this problem, we are given an edge-weighted graph~$\cG$ on $2n$ vertices
and the goal is to compute a \emph{perfect matching}---a set of $n$ disjoint edges---of 
minimum total weight. The problem, which we will simply refer to as the
\emph{matching problem}, can be solved in polynomial time using the well-known 
blossom algorithm by Edmonds~\cite{Edmonds1965MaximumMA,Edmonds_1965}.
The currently fastest algorithm is due to Gabow~\cite{gabow-matching} and runs in $O(n(m + n \log n))$ time,
where $m$ is the number of edges of~$\cG$. In the Euclidean variant of 
the problem, the input graph~$\cG$ is the complete graph on a set $P$ of $2n$ points in 
Euclidean space and the edge weights correspond to the Euclidean distances between the points.
Varadarajan~\cite{DBLP:conf/focs/Varadarajan98}, improving on work 
of Vaidya~\cite{DBLP:journals/siamcomp/Vaidya89a}, showed that the 
matching problem in $\Reals^2$ can be solved in $O(n\sqrt{n}\log^5 n)$ time.
\medskip

We are interested in the dynamic variant of the Euclidean matching problem, where pairs of
points may be inserted into or deleted from the set~$P$. One can, of course, 
compute an optimal solution from scratch after each update, but this is time consuming.
More importantly, the new solution may be completely different from the previous one. 
Since breaking up existing pairs in the matching and establishing new pairs
may be expensive, switching to a completely new solution is undesirable.
Thus, we want to maintain a perfect matching that does not change much after
each update, while still having low weight.

This is similar to the goal of online algorithms, where the input is revealed
over time---in other words, there are only insertions---and one has to maintain
a feasible solution without revoking earlier decisions. For the matching problem
as defined above, irrevocable decisions do not make much sense: the only option would 
be to always match the two newly 
arrived points to each other, which can give an arbitrarily bad approximation ratio. 
Hence, research on the online matching problem mainly focused on bipartite matching,
where the point set $P$ is partitioned into two subsets: 
a set~$S$ of servers known in advance, and a set~$C$ of
clients that arrive over time and have to be matched to the servers.
This problem is already quite difficult in~$\Reals^1$,
where it has been studied by various authors.
The currently best algorithm for the online bipartite matching in~$\Reals^1$
is by Raghvendra~\cite{raghvendra2018}; it has a competitive ratio
of $O(\log n)$, improving an earlier $O(\log^2 n)$ bound by Raghvendra and Nayyar~\cite{Nayyar-Raghvendra2017}.
Recently, Peserico and Scquizzato~\cite{line-lower-bound} proved that any
online algorithm for bipartite matching in~$\Reals^1$ must have competitive ratio
$\Omega(\sqrt{\log n})$, thus providing the first non-constant lower bound.
There are also results on online bipartite matching in $\Reals^d$ for $d>1$
and in arbitrary metric spaces; 
see for example~\cite{permutation-Kalyanasundaram,permutation-Khuller,raghvendra2016}. 
Work that is closely related to our setting is by
Gupta, Krishnaswamy, and Sandeep \cite{bicromatic-with-recourse}, who show how to maintain a
3-approximation for the bipartite matching problem in~$\Reals^1$, while changing $O(\log n)$
edges in the matching upon each update. 
(They also obtain a result for general metrics, where the approximation ratio increases to~$O(\log n)$,
and for the case where clients as well as servers can arrive or depart; in the latter scenario,
the number of changes becomes $O(\log \sigma$), where $\sigma$ is the spread of the metric.) A similar result was achieved by Megow and N\"olke~\cite{DBLP:journals/algorithmica/MegowN25}.
Recently, Goranci et al. \cite{DBLP:conf/icml/GoranciKPSSZ25} also study the closely related bipartite matching problem in any fixed-dimensional Euclidean space, where both clients and servers can arrive or depart. They show that, using randomization techniques and two quadtree-like data structures, it is possible to achieve an algorithm with an $O((1/\eps)n^\eps)$ update runtime while maintaining an expected $O(1/\eps)$-approximation for any $0 < \eps \leq 1$. Moreover, they provide a lower bound of $2$ on the approximability of an algorithm with a sublinear amortized update runtime, which is similar to one of the results presented in this paper.

Our interest, however, lies in non-bipartite Euclidean matching. 
Thus, we must allow some changes to the existing
matching when new points are inserted or deleted. In other words, using terminology
from online algorithms, we must allow \emph{recourse}, similar to the papers by
Gupta, Krishnaswamy, and Sandeep \cite{bicromatic-with-recourse} and
Megow and N\"olke~\cite{DBLP:journals/algorithmica/MegowN25} we just mentioned.
The question then becomes: how many changes are required
to maintain a solution that has a good approximation ratio?
Or, more generally: which trade-offs can we obtain between
the number of changes we are allowed to make and the approximation ratio?
This problem was also considered by Bhore, Filtser, and Toth~\cite{DBLP:conf/soda/BhoreFT24},
who give a randomized $O(\log \sigma)$-approximation algorithm with recourse~$O(\log \sigma)$
for metric spaces of bounded doubling dimension, where $\sigma$ is the spread of the input.
(They also give a deterministic algorithm with recourse~$O(\log^2 n)$ that maintains a matching of weight
$O((\log n)\cdot \mbox{\sc mst}(P))$, where $\mbox{\sc mst}(P)$ denotes the weight of 
a minimum spanning tree on~$P$. Note that $\mbox{\sc mst}(P)$ can be arbitrarily much larger
than the minimum weight of a perfect matching.)
Furthermore, they note that a 3-stable algorithm---that is, an algorithm with recourse~1---cannot 
have bounded approximation ratio, even in~$\Reals^1$.

A related problem was studied by Matuschke, Schmidt-Kraepelin, and Verschae \cite{online-metric-matching}, 
who study non-bipartite matching in the following  setting. 
Instead of considering a sequence of insertions of
point pairs, they consider a scenario with two stages. In the first stage, a set $P_1$
of arbitrary (even) size is given, for which the algorithm must compute a perfect matching~$M_1$. 
In the second stage, a set $P_2$ of $2n$ points is given, after which the algorithm must compute
a perfect matching $M_2$ for $P_1\cup P_2$. The goal is that the maximum approximation
ratio of~$M_1$ and~$M_2$, denoted by~$\alpha$, is small, while the number of edges 
removed from $M_1$ in the second stage is small. For point sets in~$\Reals^1$,
they obtain approximation ratio~$\alpha=10$ with $|M_1\setminus M_2|\leq 2n$. 
When $n$ is known in advance, they improve the result to $\alpha=3$ and $|M_1\setminus M_2|\leq n$;
this latter result actually holds in arbitrary metric spaces.

\subparagraph{Our results.} 
We consider the matching problem for a fully dynamic point set $P$, focusing 
on trade-offs between the number of changes to the matching after each update and the
approximation ratio. Following previous work on the bipartite setting,
we study this fundamental problem in $\Reals^1$, where it already appears to be quite challenging.
We will state our results using the terminology introduced by
De~Berg, Sadhukhan, and Spieksma~\cite{range-assignment,domset-indset}. To this end,
consider a dynamic algorithm~\alg that maintains a perfect matching on~$P$.
Let $P(t)$ be the point set at time~$t$, where we assume $P(0)=\emptyset$,
let $\Malg(t)$ be the matching computed by~$\alg$ on~$P(t)$,
and let $\alg(t) := \weight(\Malg(t))$ be the weight of~$\Malg(t)$.
Furthermore, let $\Mopt(t)$ be an optimal (that is, minimum-weight) perfect matching on~$P(t)$,
and let $\opt(t) := \weight(\Mopt(t))$. 

We say that \alg is an \emph{$f(n)$-stable algorithm} if, for any sequence of
updates to the set~$P$, we have $|\Malg(t+1) \symdif \Malg(t)| \leq f(n)$ 
for all~$t\geq 0$, where $n := \tfrac{1}{2}\cdot \max(|P(t)|,|P(t+1)|)$ and $\Delta$ denotes the symmetric difference.
We refer to $f(n)$ as the \emph{stability} of \alg. The stability is essentially the
same as the \emph{recourse} used for online algorithms: if $|\Malg(t+1) \symdif \Malg(t)|=k$
then the number of edges deleted from the previous matching is equal to~$(k-1)/2$.
We say that \alg is an \emph{$\alpha$-approximation algorithm}
if $\alg(t) \leq \alpha\cdot \opt(t)$ for all $t\geq 0$. 

Our first result, presented in \autoref{sec:2-approx}, is an 
$O(\sqrt{n})$-stable $2$-approximation algorithm. 
Note that, unlike the results of Bhore, Filtser, and Toth~\cite{DBLP:conf/soda/BhoreFT24},
our bounds do not depend on the spread~$\sigma$ of the instance.
We also prove that the approximation ratio~2 is best possible, not just for
$O(\sqrt{n})$-stable algorithms but for any algorithm that is $o(n)$-stable. 
This result relies on a construction that is very similar to the one used by Goranci~et~al.~\cite{DBLP:conf/icml/GoranciKPSSZ25}.
Since our result explicitly lower bounds the per-update stability rather than the amortized runtime over a sequence of updates,
we include it for completeness.
Our second result, presented in \autoref{sec:lb-stability}, is that no algorithm with bounded 
approximation ratio can have stability~$o(\log n)$. 
This extends the lower bound by Bhore, Filtser, and Toth~\cite{DBLP:conf/soda/BhoreFT24},
who show that any $O(1)$-stable algorithm has approximation ratio~$\Omega(\log n)$ (and, more generally,
that any $r$-stable algorithm has approximation ratio~$\Omega(\log n/(r\log r))$.
Our lower bound uses an interesting construction, where the locations of the points are 
given by a generating function derived from the binary representation of the arrival times.
We also show how to adapt both lower bounds to the bipartite variant of the problem, where in each event 
the pair of points being inserted or deleted contains one server and one client. In this setting, a feasible solution is a perfect matching 
such that each edge in the matching connects a server to a client.

\subparagraph{Notation and terminology.} 
It will be convenient to view the matching problem in purely geometric terms. Thus, we
are given a dynamic point set $P$ in $\Reals^1$, and we want to maintain a set of
$|P|/2$ edges---segments connecting points from~$P$---such that every point
in~$P$ is an endpoint of exactly one edge. We refer to such a set of edges as 
a \emph{matching} on~$P$, omitting the adjective ``perfect'' for convenience.
If the matching has fewer than $|P|/2$ edges, we use the term \emph{partial matching}.
We denote the optimal matching and the matching maintained by our algorithm \alg on a point set $P$ by
$\Mopt(P)$ and $\Malg(P)$, respectively; thus $\Mopt(t)$ and $\Malg(t)$ are
shorthands for $\Mopt(P(t))$ and $\Malg(P(t))$.

We denote the edge between two points
$p_i,p_j\in P$ by $p_i p_j$ and we denote its length by~$|p_i p_j|$. 
Hence, the weight of a (possibly partial) matching~$M$ is $\weight(M) := \sum_{p_ip_j\in M} |p_ip_j|$.
We denote the total length of a set $I$ of intervals by $\|I\|$, and
for a (possibly partial) matching~$M$ we sometimes use $\|M\|$ as a shorthand for $\weight(M)$.
For a (possibly partial) matching~$M$ and an interval $[x_1,x_2]$, we define 
$M\cap [x_1,x_2] := \{e\cap [x_1,x_2] :e\in M\}$ to be the parts of the edges 
from $M$ that lie in~$[x_1,x_2]$.
Note that $M\cap [x_1,x_2]$ need not be a (partial) matching; it is just a set of segments. 
We say that an edge $p_i p_j$ \emph{covers} an edge $p_k p_{\ell}$ if $p_k p_{\ell} \subseteq p_i p_{j}$.

When specifying a point set, we always list its points from left to right. Thus, if
we write $P :=\{p_1,p_2,\ldots,p_{2n}\}$ then 
 $p_1<p_2<\cdots<p_{2n}$. 
For $1\leq i<2n$, we call $p_i$ the \emph{predecessor} of $p_{i+1}$
and we call $p_{i+1}$ the \emph{successor} of $p_{i}$.
The \emph{parity} of a point $p_i$ in a set $P$ is the parity of its index. Points of odd
parity are called \emph{odd points} and points of even parity are called \emph{even points}.
Note that the parity of a point may change when $P$ is being updated.
An edge between consecutive points is called an \emph{elementary edge}.
It is easily seen that $\Mopt(P)$ is unique and
consists of elementary edges, as stated in the following observation. 
\begin{observation} \label{obs:optimal-matching}
Let $P:=\{p_1,\ldots,p_{2n}\}$. Then $\Mopt(P) = \{ p_1 p_2, p_3 p_4, \ldots, p_{2n-1} p_{2n}\}$.
\end{observation}

\section{An $O(\sqrt{n})$-stable $2$-approximation algorithm}
\label{sec:2-approx}
Our algorithm maintains a decomposition~$\B(t)$ of the current point set $P(t)$ into $O(\sqrt{n})$ 
\emph{blocks}---sets of consecutive points in~$P$---of size $O(\sqrt{n})$ each, and then 
handles the insertion (or deletion) of a pair $q_1,q_2$ as follows. 
Let $B(q_1)$ and $B(q_2)$ be the blocks containing $q_1$ and~$q_2$,
respectively. We will reconstruct the matchings in $B(q_1)$ and $B(q_2)$
from scratch and only make $O(1)$ changes to each of the other blocks.
While this idea is simple, making it work is highly non-trivial. Next we describe the
various invariants we need to maintain and prove that they guarantee
a good approximation ratio. Then we prove the existence of a so-called 
\emph{good pair of splitters} for a block, which determines the matching inside a block. Finally,
we explain how to maintain the invariants with an $O(\sqrt{n})$-stable algorithm.

\subsection{The structure}
\label{subsec:structure}

\subparagraph{The block decomposition.}
Let $P := P(t)$ be the current point set. Recall that the points from $P$ are numbered
as $p_1,\ldots,p_{2n}$, from left to right.
A \emph{block} is defined as a subset of consecutive points from~$P$, that is,
a block is a subset~$B\subset P$ of the form $\{p_i,\ldots,p_j\}$ for indices $1\leq i\leq j \leq 2n$.
We denote the leftmost point of a block~$B$ by $\lp{B}$ and the rightmost point by $\rp{B}$.
With a slight abuse of notation, and when $P$ is clear from the context, we will write 
$\Malg(B)$ to denote the edges in $\Malg(P)$ with both endpoints in the same block~$B$.
Similarly, we write $\Mopt(B)$ for the edges in $\Mopt(P)$ with both endpoints in~$B$.

A \emph{block decomposition} of $P$ is an ordered set $\B := \{B_1,B_2,\ldots,B_{|\B|}\}$
of blocks that defines a partition of~$P$, where $B_1$ is the leftmost block and $\rp{B_i}$ 
is the predecessor of $\lp{B_{i+1}}$ for all~$1\leq i<|\B|$.
The $O(\sqrt{n})$ stability parameter will be ensured by 
the following invariant, where we define $m \defeq \floor{\sqrt{|P(t)|}}$:
\begin{enumerate}[label=\textbf{(B.\arabic*)}, labelwidth=*, align=left, itemindent=0pt, labelsep=1em, leftmargin=*, widest=9]
        \item For all $B\in \B$, it holds that $m \leq |B| \leq 2m + 3$.  \label{inv:B1}
\end{enumerate}
We will need two additional technical invariants to be able to maintain Invariant~\ref{inv:B1}
in an efficient manner. We will state these later, when we describe the details of the update algorithm.

\subparagraph{Connectors between blocks.}
The matching $\Malg$ maintained by our algorithm will mainly use intra-block edges, 
that is, edges that connect points within the same block. 
Our decomposition invariant does not guarantee that blocks are of even size, however,
so we may also need some inter-block edges.
We will only use inter-block edges connecting the rightmost point of a
block~$B_i$ to the leftmost point of the next block~$B_{i+1}$. 
We call such inter-block edges \emph{connectors}.
The next invariant specifies which connectors we use and shows that 
the set of connectors used by $\Malg$ is uniquely determined by the current block decomposition. 
\begin{quotation}
\noindent \textbf{Connector invariant.} 
Let $B_i$ and $B_{i+1}$ be consecutive blocks in the current block decomposition~$\B$,
and let $p_j := \rp{B_i}$. Then the matching $\Malg$ satisfies:
\begin{enumerate}[label=\textbf{(C)}, labelwidth=*, align=left, itemindent=0pt, labelsep=1em, leftmargin=*, widest=9]
\item $p_j p_{j+1}\in \Malg$ iff $p_j$ is an odd point.  \label{inv:C}
\end{enumerate}
\end{quotation}
Note that $B_{i+1}$ must exist when $\rp{B_i}$ is an odd point. 
\autoref{obs:optimal-matching} and the connector invariant together imply that
a connector $p_j p_{j+1}$ is used in $\Malg(P)$ iff it is used in $\Mopt(P)$.

\subparagraph{Unwrapped, inner-wrapped, and outer-wrapped blocks.}
We now come to the heart of our construction, which is how we match points within a block~$B_i\in \B$.
Let $p_j := \lp{B_i}$ and $p_k := \rp{B_i}$. Recall that the connector $p_{j-1} p_j$ is 
an edge in~$\Malg$ iff $p_j$ is an even point. Similarly, $p_{k}p_{k+1}$ is an edge in $\Malg$ iff 
$p_k$ is an odd point. 
This leads us to define $\Bin_i := \{ p_a,\ldots,p_b\}$,
where $a=j$ if $j$ is odd and $a=j+1$ if $j$ is even, and 
$b=k$ if $k$ is even and $b=k-1$ if $k$ is odd. (If $|B_i| \leq 2$ then $\Bin_i = \emptyset$.) 
Thus, $\Bin_i$ contains
the points from $B_i$ that still need to be matched after adding the connectors;
see \autoref{fig:connectors}.
\begin{figure}
\centering
\includegraphics{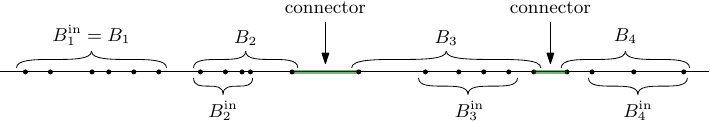}
\caption{A block decomposition.}
\label{fig:connectors}
\end{figure}
Note that $|\Bin_i|$ must be even, as it starts at an odd point and ends at an even point.
We call $\Bin_i$ the \emph{interior} of $B_i$; we may have $\Bin_i=B_i$.
Note that $\Bin_i$ may change even if $B_i$ does not change,
because the insertion of a point in a block $B_j$ with $j<i$ causes
a parity change for~$\lp{B_i}$.

We define three types of blocks, each using a specific matching in its interior.
The type of any given block---in other words, which matching $\Malg$ uses in the
interior of the block---will
depend on the parity of $\lp{B_i}$, as explained later. 
The three types are called unwrapped, inner-wrapped, and outer-wrapped.
To define these types for a block~$B_i\in\B$, we first define what it means
for a sub-block~$B\subseteq B_i$ to be wrapped or unwrapped; see \autoref{fig:wrapped-unwrapped}.
\begin{figure}
\centering
\includegraphics{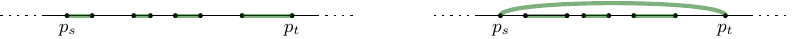}
\caption{An unwrapped sub-block (left) and a wrapped sub-block (right).}
\label{fig:wrapped-unwrapped}
\end{figure}
\begin{definition}
Let $B_{s,t}:=\{p_s,\ldots,p_t\}$ be an even-sized sub-block of a block $B_i\in \B$. 
Define 
\[
\Mun(B_{s,t}) := \{p_s p_{s+1},p_{s+2} p_{s+3},\ldots,p_{t-1}p_t\}
\]
and 
\[
\Mwr(B_{s,t}) := \{ p_s p_t\}\cup \{p_{s+1} p_{s+2},p_{s+3} p_{s+4},\ldots,p_{t-2}p_{t-1}\}.
\]
We call $B_{s,t}$ \emph{unwrapped} if $\Malg(B_{s,t}) = \Mun(B_{s,t})$ and
\emph{wrapped} if $\Malg(B_{s,t}) = \Mwr(B_{s,t})$. 
\end{definition}
Note that if a sub-block $B_{s,t}$ is unwrapped and $p_s$ is an odd point, 
then $\Malg(B_{s,t})=\Mopt(B_{s,t})$ by \autoref{obs:optimal-matching}.
The definition of an unwrapped block is straightforward.
\begin{itemize}
\item \emph{Unwrapped blocks.} An unwrapped block~$B_i$ is a block such that $\Bin_i$ is unwrapped.
\end{itemize}
Inner- and outer-wrapped blocks will be defined with respect to two points 
$\lsplit{B_i}$ and $\rsplit{B_i}$ chosen from~$\Bin_i$, which we call the \emph{splitters} of~$B_i$. 
The choice of the splitters is crucial to obtain a good approximation ratio, and it 
will be discussed in detail below. For now, it suffices to know that splitters
have the following property. Let $\Bin_i := \{p_a,\ldots,p_b\}$
and let $\lsplit{B_i} := p_{s}$ and $\rsplit{B_i} := p_{t}$. 
Then $p_a \leq p_s<p_t\leq p_b$ and $\{p_s,\ldots,p_t\}$ is an even-sized sub-block. 
Inner- and outer-wrapped blocks are now defined as follows; 
see also \autoref{fig:toggle}.
\begin{figure}[b]
\centering
\includegraphics{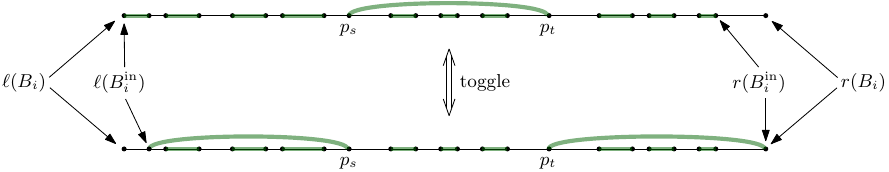}
\caption{An inner-wrapped block (top) and an outer-wrapped block (bottom).
         If the parity of $\lp{B_i}$---and, hence, of $p_s$---changes
         due to an insertion in an earlier block,
         then the type of $B_i$ changes from inner-wrapped to outer-wrapped,
         or vice versa. 
         We then say that $B_i$ is toggled. Note that in the situation at the bottom,
         the point $\lp{B_i}$ must be incident to a connector.}
\label{fig:toggle}
\end{figure}
\begin{itemize}    
\item \emph{Inner-wrapped blocks.} An inner-wrapped block is a block $B_i\in\B$ where
      $\Malg$ is such that $\{p_a,\ldots,p_{s-1}\}$ and $\{p_{t+1},\ldots,p_{b}\}$ are unwrapped 
      and $\{p_{s},\ldots,p_{t}\}$ is wrapped. 
\item \emph{Outer-wrapped blocks.} An outer-wrapped block is a block $B_i\in\B$ where
     $\Malg$ is such that $\{p_a,\ldots,p_{s}\}$ and $\{p_{t},\ldots,p_{b}\}$ are wrapped and 
     $\{p_{s+1},\ldots,p_{t-1}\}$ is unwrapped. 
\end{itemize}
Observe that $\{p_a,\ldots,p_{s-1}\}$ must have even size for an inner-wrapped
block to be well-defined. Since $p_a$ is an odd point by definition,
$\{p_a,\ldots,p_{s-1}\}$ has even size iff $\lsplit{B_i}$ is an odd point.
(Recall that $p_s=\lsplit{B_i}$.)
Similarly, $\lsplit{B_i}$ must be an even point for an 
outer-wrapped block to be well-defined. This does not mean that the type of a 
block~$B_i$ is fixed as long as $B_i$ and $\lsplit{B_i}$ do not change. Indeed, 
the parity of $\lsplit{B_i}$ can change due to an insertion into a block~$B_j$
for some~$j<i$. When this happens, we must \emph{toggle} $B_i$, that is,
we must change its type from inner-wrapped to outer-wrapped, or vice versa. 
As illustrated in \autoref{fig:toggle}, toggling a block only requires $O(1)$ 
changes in $\Malg$. (The figure illustrates the situation
where $|B_i|$ is odd, but toggling also requires 
only $O(1)$ changes when $|B_i|$ is even.) 

We want to pick $\lsplit{B_i}$ and $\rsplit{B_i}$ such that
we have a good approximation ratio regardless if $B_i$ is inner-wrapped or outer-wrapped. This leads to the following definition.
\begin{definition}\label{def:good-splitters}
Let $B_i$ be a block with $|B_i|\geq 6$ and define $\opt(\Bin_i) := \weight(\Mopt(\Bin_i))$.
Then two points $\lsplit{B_i}:= p_s$ and 
$\rsplit{B_i}:= p_t$ with $p_s<p_t$ form a \emph{good pair of splitters} if
\begin{enumerate}
\item[(i)]  The set $\{p_s,\ldots,p_t\}$ is an even-sized sub-block of $B_i$. 
\item[(ii)] If $p_s$ is an odd point in the current set~$P$ then the total weight of the 
            inner-wrapped matching of $B_i$ is at most $2 \cdot \opt(\Bin_i)$,
            and if $p_s$ is an even point in the current set~$P$ then the total weight of the 
            outer-wrapped matching of $B_i$ is at most $2 \cdot \opt(\Bin_i)$. 
\end{enumerate}
\end{definition}
Our algorithm will maintain to the following invariant.
\begin{quotation}
\noindent \textbf{Wrapping invariant.} 
Let $B_i$ be a block in the current block decomposition~$\B$. The matching $\Malg$ 
satisfies the following property for a good pair of splitters $\lsplit{B_i},\rsplit{B_i}$.
\begin{enumerate}[label=\textbf{(W.\arabic*)}, labelwidth=*, align=left, itemindent=0pt, labelsep=1em, leftmargin=*, widest=9]
\item If $|B_i|<6$ then $B_i$ is unwrapped. \label{inv:W1}
\item If $|B_i|\geq 6$ and $\lsplit{B_i}$ is an odd point then $B_i$ is inner-wrapped.  \label{inv:W2}
\item If $|B_i|\geq 6$ and $\lsplit{B_i}$ is an even point then $B_i$ is outer-wrapped.   \label{inv:W3}
\end{enumerate}
\end{quotation}
In the next subsection we show that a good pair of splitters always exists, but
first we show that the wrapping and connector invariants together 
imply that $\Malg$ is a $2$-approximation.
\begin{restatable}{lemma}{twoapprox} \label{lem:2-approx}
Let $\Malg(P)$ be a matching adhering to the wrapping invariant
and the connector invariant. Then $\Malg(P)$ is a perfect matching on $P$
such that $\weight(\Malg(P))\leq 2\cdot \opt(P)$.
\end{restatable}
\begin{proof}
$\Malg(P)$ is a perfect matching because of the connector invariant
and the fact that the wrapping invariant guarantees that the edges used in each block~$B_i$
form a perfect matching of~$\Bin_i$. To bound the total weight of $\Malg(P)$, we first
observe that the set $\C$ of connectors used by $\Malg(P)$ is also used in $\Mopt(P)$, by invariant~\ref{inv:C}.
The same is true for the edges used in unwrapped blocks, because $\Bin_i$ is an even-sized
sub-block starting with an odd point. The remaining edges used by $\Malg$ are the ones used in
inner-wrapped or outer-wrapped blocks. Their total weight is at most twice the total weight
of the edges of $\Mopt(\Bin_i)$, by the wrapping invariant and the definition of a good
pair of splitters. Hence,
\[
\begin{array}{lll}
\alg(P) & = & \sum\limits_{e\in \C} |e| + \sum\limits_{B_i\in\B} \alg(\Bin_i) \\[1mm]
                  & \leq & \sum\limits_{e\in \C} |e| + \sum\limits_{B_i\in\B} 2\cdot \opt(\Bin_i) \\[1mm]
                  & \leq & 2\cdot \left( \sum\limits_{e\in \C} |e| + \sum\limits_{B_i\in\B} \opt(\Bin_i) \right) \\[1mm]
                  & = & 2\cdot \opt(P),
\end{array}\]
which proves the lemma.
\end{proof}

\subsection{Finding a good pair of splitters}
\label{subsec:splitters}
Let $B$ be a block of size at least~$6$ for which we want to find a good pair
of splitters. It will be convenient to label the points in $B$ from left to right
as $b_0,b_1,\ldots,b_{z+1}$, but it is important to realize that there may be other points
from $P$ before~$B$. Thus, the parity of $b_0$ can still be odd or even. 
Recall that the splitters should be chosen from $\Bin$, the interior of the block~$B$.
Depending on the parity of $b_0$ (which can change), $b_0$ and $b_{z+1}$ may
or may not be part of~$\Bin$, so we have to choose the splitters from $\{b_1,\ldots,b_z\}$.
To simplify the notation, we assume that $b_1=0$ and $b_z=1$. This is without loss of generality,
as scaling the instance does not change the approximation ratio we achieve inside the block.
\medskip

Let $x_1,x_2\in [0,1]$ be a potential pair of splitters, with $x_1<x_2$.
Then $x_1$ and $x_2$ split the interval $[0,1]$ into two parts: an \emph{inner part} $[x_1,x_2]$
and an \emph{outer part}~$[0,x_1)\cup (x_2,1]$. We first find a suitable partition
into an inner and an outer part in a continuous setting, where we can pick
the splitting points $x_1$ and $x_2$ anywhere in the interval~$[0,1]$. 
After that, we explain how to obtain a good pair of splitters $b_s,b_t \in \{b_1,,\ldots,b_z\}$
from the pair~$x_1,x_2$.
\[
\mbox{Define} \hspace{5mm} \Ieven := [b_1,b_2]\cup [b_3,b_4] \cup \cdots \cup [b_{2y-1},b_{2y}], \mbox{where $y=\floor{z/2}$.} 
\]
Thus, $\Ieven$ contains the edges that are in $\Mopt\cap[0,1]$ when $b_0$ is an even point.
We now define a \emph{target length} for either the inner or the outer part.
Recall that $\|\Ieven\|$ denotes the total length of the edges in $\Ieven$;
thus, $\|\Ieven \| = \weight(\Ieven)$. 
The target length will depend on $\|\Ieven\|$ and is defined as follows.
Let $\phi:[0,1]\rightarrow [0,1]$ be given by 
\[
\phi(\ell) :=
\begin{cases}
\tfrac{\ell}{2(1-\ell)} & \mbox{ if $0\leq \ell\leq \tfrac{1}{2}$} \\[1mm]
\tfrac{3\ell-1}{2 \ell} & \mbox{ if $\tfrac{1}{2}<\ell\leq 1$} 
\end{cases}
\]
The target length is now defined $\phi(\|\Ieven\|)$. To shorten the formulas 
we define $W:= \|\Ieven\|$ so that the target length becomes $\phi(W)$.
(Unfortunately, there is no clear intuition behind the definition of $\phi(W)$; 
it is simply the function such that the computations work out.)
\begin{lemma}\label{lem:continuous}
For any block~$B$, there exists $x_1,x_2\in [0,1]$ with $x_1<x_2$ such that 
\begin{enumerate}[label=(\roman*), labelwidth=*, align=left, itemindent=0pt, labelsep=1em, leftmargin=*, widest=9]
\item $\|\Iin\| = \phi(W)$ and $\|\Iin \cap \Ieven\| = W\cdot \phi(W)$, or
\item $\|\Iout\| = \phi(W)$ and $\|\Iout \cap \Ieven\| = W\cdot \phi(W)$,
\end{enumerate}
where $\Iin := [x_1,x_2]$ and $\Iout := [0,x_1)\cup (x_2,1]$ are the inner and outer parts defined by~$x_1,x_2$. 
Moreover, a pair $x_1,x_2$ with the required property can be computed in $O(|B|)$ time. 
\end{lemma}
\begin{proof}
For $x\in[0,1]$ we define an interval, or pair of intervals, $I_x$ as follows:
\[
I_x := 
\begin{cases}
[x,x+\phi(W)] & \mbox{ if $x\leq 1-\phi(W)$} \\
[0,x-1+\phi(W)] \cup [x,1]  & \mbox{ if $1-\phi(W)<x\leq 1$}. \\
\end{cases}
\]
Note that $\|I_x||=\phi(W)$ for all~$x$. Intuitively, as $x$ increases from $0$ to $1$, 
the interval $I_x$ moves to the right until it hits the point~$1$, after which it ``wraps around'' 
and consists of two subintervals. Define $f(x) := \|I_x\cap \Ieven\|$ 
and consider the integral $\int_0^1 f(x)\,\mathrm{d}x$. Note that a point $y\in [0,1]$
is contained in $I_x\cap\Ieven$ iff $x$ is contained in the set $\Ieven-\phi(W)$, that is,
in the set obtained by shifting $\Ieven$ backwards over a distance~$\phi(W)$ and wrapping around at $x=0$ 
where necessary. Thus, each $y\in \Ieven$ contributes a total value $\phi(W)$ to the integral. Hence,
\[
\int_0^1 f(x)\,\mathrm{d}x =  \|\Ieven\| \cdot \phi(W) = W \cdot \phi(W),
\]
where the second equality follows from simply plugging in the definition of~$W$.
Thus, $\min_{0\leq x\leq 1} f(x) \leq W\cdot \phi(W) \leq \max_{0\leq x\leq 1} f(x)$, and since
$f(x)$ is continuous, this implies that there exists an $x^*$ such that $f(x^*)=W\cdot \phi(W)$.
To satisfy the conditions of the lemma we can therefore set 
$x_1:=x^*$ and $x_2 := x^*+\phi(W)$ if $x^*\leq 1-\phi(W)$, and
we can set $x_1:=x^*-1+\phi(W)$ and $x_2 := x^*$ otherwise.

Observe that $f(x)$ is a piecewise linear function with $O(|B|)$ pieces, which we can
construct in $O(|B|)$ time, after which we can find $x^*$ in $O(|B|)$ time.
\end{proof}
Let $x_1,x_2$ be a pair of points with the property stated in \autoref{lem:continuous} 
and let $b_i$ and $b_j$ such that $b_i\leq x_1< b_{i+1}$ and $b_j<x_2\leq b_{j+1}$. 
Note that $b_i, b_{i + 1}, b_j, b_{j + 1} \in \{b_1, \dots, b_z\}$. 
To be able to define the splitters in a clean way, we first prove that $i \neq j$.
\begin{lemma} \label{lem:inoteqj}
Let $x_1,x_2$ be a pair of points with the property stated in \autoref{lem:continuous} 
and let $b_i$ and $b_j$ such that $b_i\leq x_1< b_{i+1}$ and $b_j<x_2\leq b_{j+1}$. 
Then $i \neq j$.
\end{lemma}
\begin{proof}
Assume for contradiction that $i = j$. Then $b_i\leq x_1<x_2 \leq b_{i+1}$. 
Since $|B| \geq 6$ and all points in $B$ are distinct, we have $W \neq 0$ and $W \neq 1$. 
Similarly, $\phi(W) \neq 0$ and $\phi(W) \neq 1$. We distinguish two cases, 
each leading to a contradiction.
\medskip

\noindent \emph{Case~I: $x_1$ and $x_2$ satisfy property (i) from \autoref{lem:continuous}.} 
If $b_i$ is even, then $\|\Iin \cap \Ieven\| = 0$, which contradicts $W \neq 0$ and $\phi(W) \neq 0$. If $b_i$ is odd, we have $\|\Iin \cap \Ieven\| = \|\Iin\|=\phi(W)$, which contradicts that $W \neq 1$ and $\phi(W) \neq 0$.
\medskip

\noindent \emph{Case~II: $x_1$ and $x_2$ satisfy property (ii) from \autoref{lem:continuous}.} 
If $b_i$ is even, then $\|\Iout \cap \Ieven\| = \|\Ieven\|= W$, which contradicts $W \neq 0$ 
and $\phi(W) \neq 1$. If $b_i$ is odd, then $\|\Iout \cap \Ieven\| = \| \Ieven\| - \| \Iin\| = W - (1 - \phi(W))$. In the latter case 
we have $W - (1 - \phi(W)) = W \cdot \phi(W)$ by property~(ii) of \autoref{lem:continuous}.
But then $W(1 - \phi(W)) = 1 - \phi(W)$, which contradicts $W \neq 1$ and $\phi(W) \neq 1$.
\end{proof}
We now choose the splitter $b_s\in\{b_i,b_{i+1}\}$ as follows.
\begin{enumerate}[label=\textbf{(S.\arabic*)}, labelwidth=*, align=left, itemindent=0pt, labelsep=1em, leftmargin=*, widest=9]
    \item \label{S.1} If $x_1$ and $x_2$ have property~(i) from  \autoref{lem:continuous},
          we pick $b_s:=b_i$ if $\{b_0,\ldots,b_i\}$ has even size and we pick $b_s := b_{i+1}$ if it has odd size.
          Thus, $b_0$ and $b_s$ have opposite parity.
    \item \label{S.2} If $x_1$ and $x_2$ have property~(ii) from  \autoref{lem:continuous},
          we pick $b_s=b_i$ if $\{b_0,\ldots,b_i\}$ has odd size and we pick $b_s := b_{i+1}$ if it has even size.
          Thus, $b_0$ and $b_s$ have the same parity.
\end{enumerate}
After picking~$b_s$, we pick $b_t\in\{b_j,b_{j+1}\}$ such that $\{b_s,\ldots,b_t\}$ has even size.
The next lemma shows that $b_s,b_t$ is a good pair of splitters when 
$b_s$ and $b_t$ are chosen according to~\ref{S.1}; 
the case where $b_s$ and $b_t$ are chosen according to~\ref{S.2} will be handled after that.
\begin{lemma}\label{lem:S.1}
If $b_s$ and $b_t$ are chosen according to~\ref{S.1}, they are a good pair of splitters.    
\end{lemma}
\begin{proof}
The sub-block $\{b_s,\ldots,b_t\}$ has even size by construction, so it remains to 
prove property~(ii) of \autoref{def:good-splitters}.
\smallskip

\noindent \emph{Case~I: $b_s$ is currently an odd point in~$P$.} 
Let $\Min$ denote the inner-wrapped matching of~$\Bin$.
We must show that $\weight(\Min) \leq 2 \cdot \opt(\Bin)$. 
Since we picked $b_s$ according to~\ref{S.1}, we know that the parity of $b_0$ 
is opposite to the parity of~$b_s$. Hence, $b_0$ is an even point, as in \autoref{fig:Case-I}(a).
Depending on whether $|B|$ is odd or even, $\Bin$
ends at $b_{z+1}$ and we have $\opt(\Bin) = \|\Ieven\|+|b_z b_{z+1}|=W+|b_z b_{z+1}|$, 
or $\Bin$ ends at $b_{z}$ and we have $\opt(\Bin) = \|\Ieven\|=W$.
The former case is shown in \autoref{fig:Case-I}(a). For now, assume the latter case,
that is, assume that $|B|$ is even so that $\opt(\Bin) =W$.
We can then bound $\weight(\Min)$ as follows, where (as before) $\Iin := [x_1,x_2]$.
\begin{figure}
\centering
\includegraphics{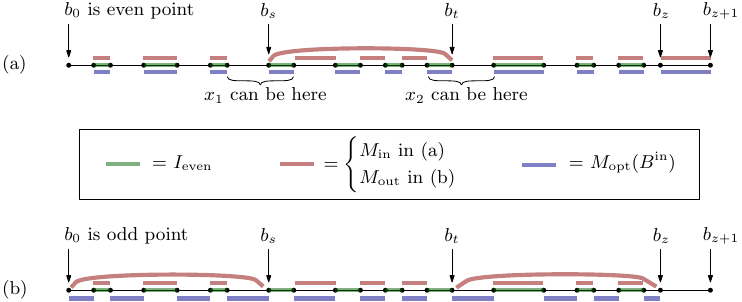}
\caption{Illustration for the proof of \autoref{lem:S.1}.}
\label{fig:Case-I}
\end{figure}
\[
\weight(\Min)
\leq   W + 2\cdot(\|\Iin\|- \|\Iin \cap \Ieven\|) 
 =   W + 2\cdot(\phi(W) - W\cdot \phi(W)),
\]
where the last equality follows from property~(i) in \autoref{lem:continuous}.
We have two cases, depending on the value of~$W$. 
\begin{itemize}
\item 
If $W\leq \tfrac{1}{2}$ then $\phi(W) = \tfrac{W}{2(1-W)}$ and so
\[
2\cdot(\phi(W) - W\cdot \phi(W)) = 2(1-W)\cdot\tfrac{W}{2(1-W)} = W.
\]
We can conclude that $\weight(\Min) \leq 2W = 2\cdot \opt(\Bin)$, as desired. 
\item
On the other hand, if $W> \tfrac{1}{2}$ then $\phi(W) = \tfrac{3W-1}{2W}$ and so
\[
2\cdot(\phi(W) - W\cdot \phi(W)) = 2(1-W)\cdot\tfrac{3W-1}{2W} = \tfrac{(1-W)(3W-1)}{W}
\]
It is easily verified that $\tfrac{(1-W)(3W-1)}{W} \leq W$ for any $0<W\leq 1$.
Hence, $\weight(\Min) \leq 2W = 2\cdot \opt(\Bin)$.
\end{itemize}
In the derivations above we assumed $|B|$ is even.
Now suppose $|B|$ is odd. Then, compared to the case where $|B|$ is even, both
$\weight(\Min)$ and $\opt(\Bin)$ increase by~$|b_z b_{z+1}|$. This decreases
the ratio $\weight(\Min)/\opt(\Bin)$, so we still
have $\weight(\Min)\leq 2\cdot \opt(\Bin)$. 
\medskip

\noindent \emph{Case~II: $b_s$ is currently an even point in~$P$.}
The proof is very similar to that of the previous case.
Let $\Mout$ denote the outer-wrapped matching of~$\Bin$. We must show that 
$\weight(\Mout) \leq 2 \cdot \opt(\Bin)$. We now have that $b_0$ must be an odd point, 
as in \autoref{fig:Case-I}(b). As before, depending on whether $|B|$ is 
odd or even, $\Bin$ can end at $b_z$---this is the case shown in \autoref{fig:Case-I}(b)---or
at $b_{z+1}$. We assume the former case; otherwise, $\weight(\Min)$ and 
$\opt(\Bin)$ both increase by~$|b_z b_{z+1}|$, which only decreases
the ratio $\weight(\Min)/\opt(\Bin)$.

If $\Bin$ ends at $b_z$ then $\opt(\Bin) = |b_0 b_1| \; + \;  1- \|\Ieven\| = |b_0 b_1| \; + \;  1-W$. 
We can now bound $\weight(\Mout)$ as follows, where $\Iout := [0,x_1)\cup (x_2,1]$.
\[
\begin{array}{lll}
\weight(\Mout)
 & \leq & |b_0 b_1| \; + \; (1-W) \; + \; 2\cdot\|\Iout \cap \Ieven\| \\[1mm]
 & = & |b_0 b_1| + (1-W) + 2\cdot (\|\Ieven\| - \|\Iin\cap \Ieven\|) \\[1mm]
 & = & |b_0 b_1| \; + \; (1-W) \; + \; 2\cdot  (W- W\cdot\phi(W)) \\[1mm]
 & = & |b_0 b_1| \; + \; (1-W) \; + \; \tfrac{2W(1-\phi(W))}{1-W} \cdot (1-W)
\end{array}
\]
Hence, it suffices to show that $\tfrac{2W(1-\phi(W))}{1-W}\leq 1$.
\begin{itemize}
\item 
If $W\leq \tfrac{1}{2}$ then $\phi(W) = \tfrac{W}{2(1-W)}$ and so
\[
\tfrac{2W(1-\phi(W))}{1-W} = \tfrac{2W}{1-W}\left(1-\tfrac{W}{2(1-W)}\right) = \tfrac{W(2-3W)}{(1-W)^2} \leq 1.
\]
\item
On the other hand, if $W> \tfrac{1}{2}$ then $\phi(W) = \tfrac{3W-1}{2W}$ and so
\[
\tfrac{2W(1-\phi(W))}{1-W} = \tfrac{2W}{1-W}\left(1-\tfrac{3W-1}{2W}\right) = 1.
\]
\end{itemize}
We conclude that the pair $b_s,b_t$ has the required property in all cases.
\end{proof}
The next lemma handles the case where $b_s$ and $b_t$ are chosen according to~\ref{S.2}.
\begin{lemma}\label{lem:S.2}
If $b_s$ and $b_t$ are chosen according to~\ref{S.2}, they are a good pair of splitters.    
\end{lemma}
\begin{proof}
The proof is symmetric to the proof of \autoref{lem:S.1}, but we give it for completeness.
The sub-block $\{b_s,\ldots,b_t\}$ has even size by construction, so it remains to 
prove property~(ii) of \autoref{def:good-splitters}.
\smallskip

\noindent \emph{Case~I: $b_s$ is currently an odd point in~$P$.} 
Let $\Min$ denote the inner-wrapped matching of~$\Bin$ and let $\opt(\Bin) := \weight(\Mopt(\Bin))$.
We must show that $\weight(\Min) \leq 2 \cdot \opt(\Bin)$. 
Since we picked $b_s$ according to~\ref{S.2}, we know that $b_0$ has the
same parity as $b_s$.
Thus, $b_0$ is an odd point and $\opt(\Bin) = |b_0 b_1| + 1- \|\Ieven\|=|b_0 b_1| +1-W$. 
(Here we assume, as in the proof of \autoref{lem:S.1},
that $\Bin$ ends at $b_{z}$, since otherwise the ratio $\weight(\Min)/\opt(\Bin)$ is even smaller.)
We bound $\weight(\Min)$ as follows.
\begin{align*}
  \weight(\Min) & \ \leq \ |b_0 b_1| + (1-W) + 2\cdot  \|\Iin \cap \Ieven\| \\
                & \ = \ |b_0 b_1| + (1-W) + 2\cdot (\|\Ieven\| - \|\Iout\cap \Ieven\|) \\
                & \ = \ |b_0 b_1| + (1-W) + 2 \cdot (W - W\cdot \phi(W)) \tag*{\llap{(by \autoref{lem:continuous}(ii))}}\\
                & \ = \ |b_0 b_1| + (1-W) + \tfrac{2W(1-\phi(W))}{1-W} (1-W) 
\end{align*}
Thus, it remains to show that $\tfrac{2W(1-\phi(W))}{1-W} \leq 1$, which we already did in Case~II of \autoref{lem:S.1}.

\smallskip

\noindent \emph{Case~II: $b_s$ is currently an even point in~$P$.}
Let $\Mout$ denote the outer-wrapped matching of~$\Bin$. We must show that 
$\weight(\Mout) \leq 2 \cdot \opt(\Bin)$. We now have that $b_0$ must be an even point, 
and (assuming again that $\Bin$ ends at $b_{z}$) we have
$\opt(\Bin) = \|\Ieven\|=W$.
We can bound $\weight(\Mout)$ as follows.
\[
\begin{array}{lll}
\weight(\Mout)
 & \leq & W + 2\cdot(\|\Iout\| - \|\Iout \cap \Ieven\|)  \\
 & = & W + 2\cdot  ( \phi(W) - W \cdot \phi(W))  
\end{array}
\]
It remains to show that $2\cdot(\phi(W) - W\cdot \phi(W)) \leq W$,
which we already did in Case~I of \autoref{lem:S.1}.
\end{proof}
The following lemma summarizes the main result of this subsection.
\begin{restatable}{lemma}{splitters} \label{lem:splitters}
Let $B$ be a block with $|B|\geq 6$. Then a good pair of splitters for $B$
exists and can be computed in~$O(|B|)$ time.
\end{restatable}
\begin{proof}
Lemmas~\ref{lem:S.1} and \ref{lem:S.2} together imply that 
the pair $b_s,b_t$ defined above is a good pair of splitters in all cases. 
By \autoref{lem:continuous} we can compute the points~$x_1,x_2$ in $O(|B|)$,
which immediately implies that we can compute $b_s,b_t$ in $O(|B|)$ time as well.
\end{proof}

\subsection{Maintaining the structure}
\label{subsec:maintenance}
The idea behind the update algorithm is simple: we update the blocks containing the
two inserted or deleted points from scratch and toggle the other blocks. Since the
size of a block is $O(\sqrt{n})$ and the number of blocks is $O(\sqrt{n})$,
and toggling a block changes only $O(1)$ edges, this means that the update
algorithm is $O(\sqrt{n})$-stable. 
Making this idea work is not straightforward, however, and requires two additional invariants
on the block decomposition. Below we describe the details.

\subparagraph{Balanced block decompositions.}
For a given non-empty set $P$ of points, with $n \defeq |P|$ and $m \defeq \floor{\sqrt{n}}$, we call a block decomposition $\B$ of $P$ \emph{balanced} if it satisfies the following three conditions, where the latter two are used in order to bound the number of blocks of the smallest and of the largest size allowed by the first condition. This will enable the algorithm to efficiently react in case $m$ changes after an event:
    \begin{enumerate}[label=\textbf{(B.\arabic*)}, labelwidth=*, align=left, itemindent=0pt, labelsep=1em, leftmargin=*, widest=9]
        \item For all $B\in \B$, it holds that $m \leq |B| \leq 2m + 3$.  \label{inv:B1}
        \item It holds that $k_1 \leq (m + 1)^2 - n - 1$, where $k_1$ is the number of blocks of size $m$ in $\B$. \label{inv:B2}
        \item Similarly, $k_2 \leq n - m^2$, where $k_2$ is the number of $B \in \B$ with $|B| \geq 2m + 2$. \label{inv:B3}
    \end{enumerate}
We describe how to update the matching, and show how to efficiently rebalance the block decomposition when a pair of points is inserted into or deleted from~$P$, so that all invariants are maintained. While our focus is on the stability of the update algorithm, we will also discuss its running time.

To implement our update algorithm, we keep an ordered list containing
the blocks in the current block decomposition~$\B(t)$ and we maintain the following information
for each block~$B_i\in \B(t)$:
a sorted list of all points in~$B_i$, the size of~$B_i$, the parity of $\lp{B_i}$, 
the status of $B_i$ (unwrapped, inner-wrapped, or outer-wrapped) and $\Malg(\Bin_i)$ 
(the matching used in the interior of~$B_i$). We also maintain a list $\C$ of connectors that $\Malg$ uses.

\subparagraph{The update algorithm.}
Suppose we want to insert or delete a pair of points~$q_1,q_2$ into or from
the current point set~$P(t)$. Thus, $P(t+1) := P(t)\cup \{q_1,q_2\}$ or
$P(t+1) := P(t)\setminus \{q_1,q_2\}$.
Let $\B(t)$ be a balanced block decomposition, where initially $\B(0) = \emptyset$, and assume that $\Malg(t)$ adheres to the connector invariant and the wrapping invariant. 
\medskip

If $|P(t+1)| \leq 2$ then we make $\B(t+1)$ 
into a single block (or $0$ blocks), and match the two points in the block when $|P(t+1)| = 2$. Clearly this satisfies the wrapping invariant.

If $|P(t+1)| > 2$ we update $\B$ and $\Malg$ in three steps.
\begin{enumerate}
\item \label{step1}
      We update the block decomposition, without yet updating
      the matching, as follows. We first insert or delete $q_1$. Let $n$ be the number of points remaining after the insertion or deletion of $q_1$ from $P(t)$. Let $m \defeq \floor{\sqrt{n}}$. Depending on whether $q_1$ is inserted or deleted, we do the following:

      \noindent \emph{Case~I: $q_1$ is inserted.} We add $q_1$ to a suitable block~$B_i$, namely the block such that $\lp{B_i} < q_1 < \rp{B_i}$ or, if there is no such block, the block closest to~$q_1$. If $|B_i| \geq 2(m + 1)$ after the insertion then we split $B_i$ into two blocks, one of size $\floor{|B_i| / 2}$ and the other of size $\ceil{|B_i| / 2}$. Afterwards, if there exists any block $B_j$ in our new block decomposition with $|B_j| \leq m$, then let $B_k$ be a neighboring block to $B_j$, which must exist since $m < n$ (which in turn holds since $n \geq 3$). We merge $B_j$ and $B_k$ into a single block, and then if the new block contains at least $2(m + 1)$ points, we split this new block in the same manner as we split $B_i$. 
      Note that this is only done for one block $B_j$, even when there are many of size at most $m$.
      
      \noindent \emph{Case~II: $q_1$ is deleted.} This case follows symmetrically to case~I. Let $B_j$ be the block that contains $q_1$. Remove $q_1$ from $B_j$, and then if after the deletion $|B_j| \leq m$, we merge $B_j$ with a neighboring block $B_k$ and split the resulting block if the size of the resulting block is at least $2(m + 1)$, just as we did in case~I. After, if there exists any block $B_i$ in our new block decomposition with $|B_i| \geq 2(m + 1)$, split it into two smaller blocks, as we split $B_i$ in case~I.

      After adding or removing $q_1$ as described above, we add or remove~$q_2$ in the same manner (but note that $n$ will differ by $1$).
      As we will show in \autoref{lem:buffer-rebalancing}, the resulting block decomposition~$\B(t+1)$ is balanced.
      So, $\B(t+1)$ contains $O(\sqrt{n})$ blocks, each of size~$O(\sqrt{n})$.
\item \label{step2} 
      Next, we update for each block $B_k$ the parity information for the point~$\lp{B_k}$.
      Since we maintain the blocks in an ordered list and we know the size of each block, 
      this takes $O(\sqrt{n})$ time. We then update the list~$\C$ of connectors, thereby
      re-establishing the connector invariant~\ref{inv:C}.
      This can be done in $O(\sqrt{n})$ time and changes $O(\sqrt{n})$ edges in $\Malg$.
\item \label{step3}
      Finally, we update the matching of the block interiors and re-establish the wrapping invariants.
      Note that Step~\ref{step1} resulted in~$O(1)$ changed or new blocks; all other
      blocks in $\B(t+1)$ contain the same points as they did in $\B(t)$. 
      \begin{itemize}
      \item
      We construct the matching in each changed or new block
      from scratch, according to the wrapping invariants~\ref{inv:W1}--\ref{inv:W3}.
      \autoref{lem:splitters} guarantees that we can do so in $O(\sqrt{n})$ time.
      Clearly, we modify $O(\sqrt{n})$ edges in $\Malg$ in this step.
      \item
      Next, we consider the blocks whose point set did not change. Even though the point set
      of such a block~$B_k$ did not change, we may still have to change its matching in order
      to satisfy the wrapping invariant. More precisely, if $|B_k|\geq 6$ and the parity of $\lp{B_k}$
      changes, then the parity of $\lsplit{B_k}$ changes as well. When this happens we need
      to toggle $B_k$---that is, change $B_k$ from being inner-wrapped to being outer-wrapped,
      or vice versa---to satisfy the wrapping invariant. Fortunately, this involves changing only $O(1)$
      edges in the matching of~$B_k$, as discussed before and illustrated in \autoref{fig:toggle}. 
      Updating the matching in blocks $B_k$ with $|B_k|<6$ obviously changes only
      $O(1)$ edges as well.
      \end{itemize}
\end{enumerate}
This finishes the description of the update algorithm. 

\subparagraph{Correctness proof.}
It is clear from the description that the algorithm runs in $O(\sqrt{n})$
time and that it changes only $O(\sqrt{n})$ edges in the matching,
provided all invariants are maintained.
For the connector and wrapping invariants this is by construction. It remains to prove that
the rebalancing method used in the algorithm guarantees
that there are $O(\sqrt{n})$ blocks of size $O(\sqrt{n})$ each at all times. 
More precisely, we need to show that the three block invariants~\ref{inv:B1}--\ref{inv:B3} are maintained.
Intuitively, this holds due to the additional block merge during a point insertion
(and the block split during a point deletion), which turns out to be sufficient even 
when the threshold~$m$ changes. 
To see this, suppose that $m$ has increased after a point insertion. 
Since $m = \floor{\sqrt{n}}$, this implies that $n$ is a square number.
Then it must hold that there are $O(\sqrt{n})$ blocks of the lowest size, which 
is $m$ by definition. Now
imagine a buffer containing the blocks of size~$m$, which the algorithm implicitly maintains. 
By merging a block of size $m$ after each insertion, the size of this buffer decreases. 
Because $\sqrt{n}$ is a concave function, an additional $\Omega(\sqrt{n})$ point insertions must occur for $m$ to increase again.
By that moment, the buffer must have been emptied and the sizes of the blocks are adjusted to the new threshold~$m$. 
The next lemma formalizes this argument.
\begin{lemma}
    \label{lem:buffer-rebalancing}
     The update algorithm described above maintains the block invariants~\ref{inv:B1}--\ref{inv:B3}.
     Thus, the block decomposition $\B(t)$ maintained by the algorithm is balanced at all times.
\end{lemma}

\begin{proof}
    Consider some arbitrary input $P(0), \dots, P(n)$.
    We prove by induction on $t$ that for any time $t$ with $|P(t)| > 0$, it holds that $\B(t)$ is balanced. For $t = 0$, we have $|P(0)| = 0$, in which case it can easily be verified that conditions \ref{inv:B1}--\ref{inv:B3} hold. We now consider $t > 0$. First, if $|P(t)| = 2$, the algorithm assigns $\B(t) = \{P(t)\}$. In this case, $n = 2$, $m = 1$, $k_1 = 0$, and $k_2 = 0$. Again, it is easily verifiable that \ref{inv:B1}--\ref{inv:B3} hold. We consider the case where $|P(t)| > 2$. By the induction hypothesis we have that $\B(t - 1)$ is balanced. Consider that $q_1, q_2$ is the pair of points inserted in or deleted from $P(t - 1)$, with $q_1$ being the first point handled by the update algorithm. Let $P$ be the set of points after the insertion or deletion of $q_1$, and $\B$ be the block decomposition of $P$ maintained by the update algorithm after handling point $q_1$ in Step~\ref{step1}. We prove that $\B$ is balanced. Let $n \defeq |P|$ and $m \defeq \floor{\sqrt{n}}$. Let $k_1$ be the number of blocks in $\B$ of size $m$. Let $k_2$ be the number of blocks in $\B$ of size at least $2m + 2$. We analogously define the numbers $n', m', k_1',$ and $k_2'$ for $P(t - 1)$ and $\B(t - 1)$.

    We first consider the case where $m = m'$. Note that regardless whether $q_1$ is inserted or deleted, since condition \ref{inv:B1} holds for $\B(t - 1)$, any newly created block in $\B$, that is, any block in $B \in \B$ that is not in $\B(t - 1)$, always satisfies $m + 1 \leq |B| \leq 2m + 1$. 
    So $k_1 \leq k_1'$ and $k_2 \leq k_2'$, which implies that \ref{inv:B1} holds for $\B$, \ref{inv:B2} holds for $\B$ if $q_1$ was deleted, and \ref{inv:B3} holds for $\B$ if $q_1$ was inserted. Moreover, if $k_1' > 0$ we ensure that during an insertion a block of size $m$ is removed, which means that either $k_1 = 0$ or $k_1 \leq k_1' - 1 \leq (m + 1)^2 - n - 1$ (by \ref{inv:B2} for $\B(t - 1)$). Hence, \ref{inv:B2} also holds for $\B$ if $q_1$ inserted. Similarly, if $k_2' > 0$ then during a deletion, a block of size at least $2m + 2$ is removed, which similarly implies that \ref{inv:B3} also holds for $\B$ if $q_1$ was deleted.

    We now consider $m \neq m'$. We separately consider the cases where $q_1$ is inserted and when $q_1$ is deleted.
    
    \begin{itemize}
    
        \item \emph{Case~I: $q_1$ is inserted.} In this case, we must have that $m = m' + 1$ and $n = m^2$, due to $n = n' + 1$. Since \ref{inv:B2} holds for $\B(t - 1)$, we have $k_1' \leq (m' + 1)^2 - n' - 1 = m^2 - n = 0$. Using this and the fact that \ref{inv:B1} holds for $\B(t - 1)$ we get that for any block $B \in \B(t - 1)$ it holds that $m' + 1 \leq |B| \leq 2m' + 3$, so also $m \leq |B| \leq 2m + 1$. As discussed before, blocks $B \in\B \setminus\B(t - 1)$ must satisfy $m + 1 \leq |B| \leq 2m + 1$. Thus for all blocks $B\in\B$ it holds that $m \leq |B| \leq 2m + 1$, from which we see that both \ref{inv:B1} and \ref{inv:B3} hold for $\B$. In order to prove that \ref{inv:B2} also holds for $\B$, we are no longer able to use the induction hypothesis, as we do not get any information about the number $k_1$ of blocks of size $m$ from it. Despite this, we can give the following simple upper-bound on $k_1$: since $n = m^2$, then it must hold that $k_1 \leq m$. Moreover, \[(m + 1)^2 - n - 1 = m^2 + 2m + 1 - m^2 - 1 = 2m \geq k_1,\] so \ref{inv:B2} holds for $\B$.
        
        \item \emph{Case~II: $q_1$ is deleted.} In this case, we must have that $m = m' - 1$ and $n' = (m')^2$, due to $n = n' - 1$. Since \ref{inv:B3} holds for $\B(t - 1)$, we have $k_2' \leq n' - (m')^2 = 0$. Using this, we can show in the same manner as we did in Case~I that for all blocks $B$ in $\B$ it holds that $m + 1 \leq |B| \leq 2m + 3$, from which we see that both \ref{inv:B1} and \ref{inv:B2} hold for $\B$. To prove \ref{inv:B3}, we see that $n < (m + 1)^2$, thus it must hold that $k_2 \leq m + 1$. Additionally, we have $n - m^2 = (m + 1)^2 - 1 - m^2 = 2m \geq k_2$, so \ref{inv:B3} holds for $\B$.
    \end{itemize}

    In the same manner, given that $\B$ is balanced, we can prove that after inserting or deleting $q_2$ from $P$ we get that $\B(t)$ is balanced, since the algorithm handles $q_2$ identically. This concludes the proof.
\end{proof}
Having proved the correctness of our algorithm, we can now state the following result,
\begin{theorem}
\label{thm:two-approx-lb}
There is an $O(\sqrt{n})$-stable $2$-approximation algorithm for the dynamic minimum-weight 
perfect matching problem in~$\Reals^1$. The running time of the algorithm is $O(\sqrt{n})$
per insertion or deletion of a point pair.
\end{theorem}
\subparagraph{\rm \emph{Remark.}}
Instead of inserting a pair of points or deleting a pair of points, we can also handle
updates involving the insertion of one point, $q_1$, and the deletion of another point,~$q_2$. 
This can either be done directly, or by first inserting $q_1$ plus a dummy point
and then deleting $q_2$ and the dummy point.

\subsection{A lower bound on the approximation ratio} \label{subsec:lb-approx}
In the following, we show that our $O(\sqrt{n})$-stable $2$-approximation algorithm is optimal in terms of its approximation ratio among all algorithms with sublinear stability. 
\begin{theorem}
    \label{thm:2-approx-lower-bound}
    For no $\alpha<2$, there is an algorithm for the dynamic minimum-weight perfect 
    matching problem in $\mathbb{R}^1$ that is $o(n)$-stable and $\alpha$-approximate. 
\end{theorem}
\begin{proof}
    Let $\alpha<2$ and assume for contradiction the existence of an $\alpha$-approximate 
    algorithm~\alg  that is  $f(n)$-stable for some $f(n) = o(n)$. Choose $n \in \mathbb{N}$ large enough such that $f(n) \leq (1-\frac{\alpha}{2})n$, and  
    suppose that at time $t=n-1$ we have $P(t) = \{1,2,\ldots,2(n-1)\}$ and that $P(t+1) = \{1,2,\ldots,2(n-1)\} \cup \{q_1,q_2\}$, where $q_1 \defeq 0$ and $q_2 \defeq 2n-1$; 
    see \autoref{fig:2-approx-lower-bound}. 
    \begin{figure}
    \centering
    \includegraphics{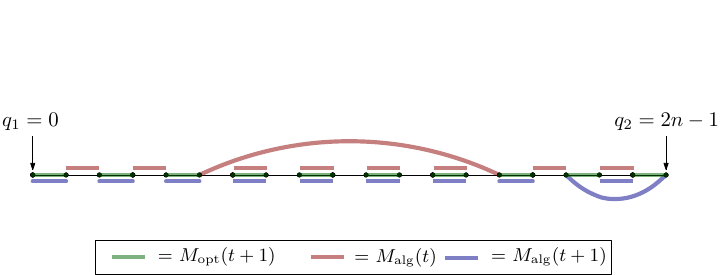}
    \caption{Counterexample used in \autoref{thm:2-approx-lower-bound}.}
    \label{fig:2-approx-lower-bound}
    \end{figure}
    From the stability of \alg and because any edge in $P(t+1)$ has weight at least $1$, we have
    \begin{align*}
    \weight(\Malg(t) \cap \Malg(t+1)) & \ \geq \ |\Malg(t) \cap \Malg(t+1)| \\[1mm]
       & \ = \ |\Malg(t) \cup \Malg(t+1)| - |\Malg(t) \,\Delta\, \Malg(t+1)| \\[1mm]
       & \ \geq \  n - f(n) \\[1mm]
       & \ \geq \ \tfrac{\alpha}{2} n.
    \end{align*}
    Since $\Malg(t+1)$ matches all points in $P(t+1)$ while $\Malg(t)$ matches all points except $\{q_1,q_2\}$, 
    the graph induced by the set $\Malg(t) \,\Delta\, \Malg(t+1)$ contains exactly two nodes with odd degree, namely $q_1$ and $q_2$. Thus, there exists a path $\pi \subseteq \Malg(t) \,\Delta\, \Malg(t+1)$ connecting $q_1$ and $q_2$.
    Clearly, $\weight(\pi) \geq 2n-1$. 
    Hence,
    \begin{align*}
        \alg(t) + \alg(t + 1)
        & \ = \ \weight(\Malg(t) \Delta \Malg(t+1)) + 2 \cdot \weight(\Malg(t) \cap \Malg(t+1)) \\
        & \ \geq \ 2n - 1 + \tfrac{\alpha}{2} (2n - 1) \\
        & \ = \ \left( 1 + \tfrac{\alpha}{2} \right) (2n - 1) \\
        & \ > \ \alpha \cdot (\opt(t) + \opt(t + 1)),
    \end{align*}
thus  $\alg(t) > \alpha \cdot \opt(t)$ or $\alg(t + 1) > \alpha \cdot \opt(t + 1)$,
giving the desired contradiction.
\end{proof}
\subparagraph{\rm \emph{Remark.}}
In the bipartite setting we can use the same construction, with points at even coordinates being servers and points at odd coordinates being clients. The algorithm is strictly more restricted in terms of the matchings it may maintain, whereas the weight of the optimal solution stays the same. Therefore \autoref{thm:2-approx-lower-bound} also holds in the bipartite setting.


\section{A lower bound on the stability} 
\label{sec:lb-stability}

In \autoref{subsec:lb-approx} we showed that any algorithm with a non-trivial upper bound on the stability must have an approximation ratio of at least $2$. We now ask the reverse question: Is there a lower bound on the stability of any algorithm that has a bounded approximation guarantee? We call the approximation guarantee of an algorithm \emph{bounded} if there exists some $\alpha:\mathbb{N}\to \mathbb{R}_{\ge 1}$ such that the algorithm is an $\alpha(n)$-approximation,
that is, such that at any time~$t$ we have $\alg(t)\leq \alpha(n)\cdot\opt(t)$, where $n := \tfrac{1}{2}|P(t)|$.

A simple example shows that the unique $1$-stable algorithm has unbounded approximation.
Assume for contradiction that it is an $\alpha(n)$-approximation for some function~$\alpha$. Consider the point set $P \defeq \{0, 1\}$ and the newly arriving points $q_1 \defeq \eps$ and $q_2 \defeq 1 + \eps$ for some $\eps < \frac{1}{\alpha(2)}$. The optimal matching has weight $2\eps$ but by $1$-stability, $\alg$ has to return the matching $\{(0,1),\, (\eps, 1+\eps)\}$, which has total weight $2 > \alpha(2) \cdot 2\eps$.

The idea can be generalized to algorithms with higher stability. Doing so requires more than two events, however, since there exists a $3$-approximate and $3$-stable algorithm for the two-stage version of our problem~\cite{online-metric-matching}. Suppose $\alg$ is $3$-stable and $\alpha(n)$-approximate, 
and consider the scenario in \autoref{fig:3stable-counterexample}, where $\eps$ is chosen sufficiently small. At time $t=2$ we have $\opt=2\eps$, forcing $\alg$ to select both edges of weight~$\eps$. At time $t=3$, the parity of the points requires $\alg$ to select at least one edge from the leftmost three points to the rightmost three points. Finally, at time $t=4$, $\alg$ can only select edges of weight $\eps^2$. Consequently, between time $t=2$ and time $t=4$, $\alg$ must delete at least three edges and insert at least four edges. Thus, the total number of modifications is at least seven, whereas a $3$-stable algorithm is permitted to make at most six changes in total. (The construction just described is essentially the same as
the one used by Bhore, Filtser, and Toth~\cite{DBLP:conf/soda/BhoreFT24} to show that an algorithm with recourse~1
cannot have bounded approximation ratio.)
In the remainder of this section we generalize this idea to obtain a counterexample for any algorithm with an $o(\log n)$ stability, leading to the following result. 
\begin{figure}
    \centering
    \includegraphics[scale=1]{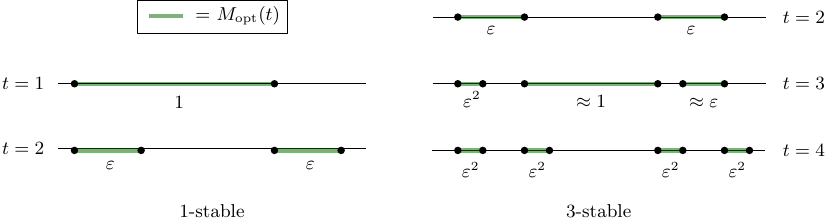}
    \caption{The counterexample for the $1$-stable algorithm is $P(1)=\{0,1\}$ with additional arrivals at $\{\eps,1+\eps\}$. For the $3$-stable algorithm, the example starts with $P(2)=\{0,\eps,1,1+\eps\}$ and adds points at $\{\eps^2,1+\eps^2\}$ at time $t=3$ and $\{\eps + \eps^2, 1 + \eps + \eps^2\}$ at time $t=4$.}
    \label{fig:3stable-counterexample}
\end{figure}

\newcommand{\depth}[1]{\ensuremath{\delta(#1)}}

\begin{theorem}
    Let \alg be an $o(\log n)$-stable algorithm 
    for the dynamic minimum-weight perfect matching problem in~$\Reals^1$. 
    Then the approximation ratio of \alg cannot be bounded as a function of~$n$. \label{thm:stability-lower-bound} 
\end{theorem}
Let \alg be an algorithm that is an $\alpha(n)$-approximation for some function 
$\alpha: \mathbb{N}\rightarrow \mathbb{R}$. Assume for contradiction that \alg is also 
$f(n)$-stable for some function $f(n) = o(\log n)$. We can assume without loss of generality 
that $f(n)$ is non-decreasing\footnote{If $\alg$ is $f(n)$-stable  
then it is also $g(n)$-stable for $g(n) = \max_{i \leq n} f(i)$, which is non-decreasing.} 
and there exists some constant $n_0$ that is a power of~$2$ 
and such that $f(n) < \log_2 n$ for all $n \geq n_0$. 
We present a sequence of $n_0$ events such that the total number of changes that \alg makes when processing this sequence is at least $n_0\log n_0$, which contradicts that \alg is $f(n)$-stable as $\sum_{t=1}^{n_0} f(t) < n_0 \log_2 n_0$. 

\subparagraph*{The construction.} 
For a positive integer $k$, define $[k] := \{1,\ldots,k\}$. 
Let $\alpha^* := \max_{t \in [n_0]}\alpha(t)$ and $\eps := \frac{1}{2\alpha^*n_0}$. Given $k \in \mathbb{N}_0$, we define
$k(i) \in \{0, 1\}$ to be the $i$-th least significant bit in the binary representation 
of~$k$. So $k = \sum_{i = 0}^{\msbp{k}} k(i) \cdot 2^i$, where $\msbp{k} := \floor{\log_2{k}}$
is the most significant bit position for $k>0$ and $\msbp{0} \defeq 0$. The least significant bit position $\lsbp{k}$ of $k$ 
is the smallest $i$ such that $k(i) = 1$. Moreover, let $p(k) \defeq \sum_{i = 0}^{\msbp{k}} k(i) \cdot \eps^i$. 
The problem instance consists of $n_0$ insertions of pairs of points,
where at time $t \in [n_0]$ we insert the points $p(2t-2)$ and $p(2t-1)$. 
Thus, for any $t$ it holds that $P(t) = \left\{p(k) \mid k \in \left\{0,\dots, 2t-1\right\}\right\}$. 
\begin{figure}
    \centering
    \includegraphics{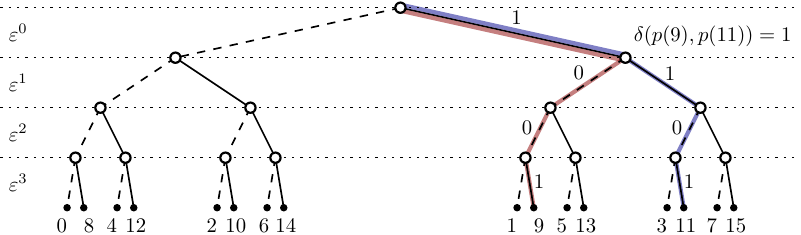}
    \caption{Trie interpretation for $P(8)$. Leaf labels correspond to times of insertion (e.g., at $t=1$ indices $\{0,1\}=\{2t-2,2t-1\}$ are inserted). Leaves map to their position on the line, e.g., the leaf $9$ is located at $p(9)=\eps^0+\eps^3$.}
    \label{fig:stability-counterexample}
\end{figure}
\subparagraph{Trie interpretation.} 
To gain intuition for our construction, it can be insightful to interpret the set $P(t)$ 
as the leaves of a binary trie, which we illustrate for $t=8$ in \autoref{fig:stability-counterexample}. 
Each level in the trie corresponds to a bit position in the binary representation of numbers in $\{0,\dots,2t-1\}$. 
Following the left edge of a node indicates the corresponding bit is~$0$, following the right indicates the bit is~$1$.
Given any $k<2t-1$, we identify the corresponding leaf by following the bits in the binary representation of $i$, 
from least to most significant. For example, $11 = (1011)_2$, so we take the path $1101$ to reach $p(11)$. 
The corresponding point on the line is located at $p(11) = \eps^0 + \eps^1 + \eps^3$. 
It is important to note that, perhaps counterintuitively, the least significant bit in the 
binary representation corresponds to the coefficient of the most significant exponent, namely $\eps^0$. 
The left-to-right ordering of the leaves in the trie corresponds to the ordering of the points in $P(t)$ on the line. 
\medskip

We define the \emph{depth} of an edge $e := p(k)p(\ell)$ 
to be $\depth{e}: = \lsbp{k \xor \ell}$, 
where $\xor$ is the bitwise XOR-operator. 
(Note that $e$ is an edge between two of the points, not an edge in the trie.)
In our trie representation, $\depth{e}$ corresponds to the level of the 
lowest common ancestor of the two leaves $p(k)$ and $p(\ell)$. 
We define $E^d$ to be the set of all edges of depth $d$. 
The following lemma relates the depth of an edge to its weight.
\begin{lemma} \label{lem:lb:edgebound}
    For any edge $e$ it holds that 
    $\weight(e) \in \left(\frac{1}{2}\eps^{\depth{e}}, \frac{3}{2}\eps^{\depth{e}}\right)$. 
\end{lemma}
\begin{proof}
Let $S \defeq \sum_{i = 1}^{\log_2(n_0)} \eps^i$. By definition we have $\eps < \frac1{3}$, 
so by the geometric-series formula we have $S < \tfrac\eps{1 - \eps} < \tfrac1{2}$. Recall that $\depth{e}$ corresponds to the smallest 
exponent for which the two polynomials corresponding to $p(k)$ and $p(\ell)$ differ. 
Assume wlog that $k(\depth{e})=1$ and $\ell(\depth{e})=0$. Then 
\[
p(k)-p(\ell) \geq \varepsilon^{\depth{e}} - \sum_{i =\depth{e}+1}^{\log_2(n_0)} \varepsilon^i \geq \epsilon^{\depth{e}}(1 - S)
\]
and 
\[
p(k)-p(\ell) \leq \sum_{i =\depth{e}}^{\log_2(n_0)}\varepsilon^i \leq \varepsilon^{\depth{e}}(1+S).
\]
Hence,
\[
|p(k) - p(\ell)| \in \left[\eps^{\depth{e}}\left(1 - S\right), \eps^{\depth{e}}\left(1 + S\right)\right] 
                \subset \left(\tfrac1{2} \eps^{\depth{e}}, \tfrac3{2} \eps^{\depth{e}}\right),
\]
which proves the claim.
\end{proof}

Note that \autoref{lem:lb:edgebound}, together with the fact that $\eps<\tfrac{1}{3}$, implies
that any edge of depth $d$ is longer than any edge of depth~$d + 1$. Besides the depth of an edge, 
we introduce the notion of the \emph{depth at time $t$}, defined as $\depth{t} \defeq \lsbp{t}$. 
We show that there is a close link between $\delta(t)$ and the matchings $\Mopt(t)$ and $\Malg(t)$. First, we give an upper bound on $\opt(t)$.

\begin{lemma}\label{lem:lb:optbound}
    For any $t \in [n_0]$ it holds that $\opt(t) \leq \eps^{\delta(t)} t$.
\end{lemma}

\begin{proof}
    We prove the lemma by constructing a perfect matching $M$ of weight $\eps^{\delta(t)} t$. Let $m \defeq 2^{\depth{t}}$. Define $M \defeq \{p(k) p(k \xor m) \mid 0 \leq k < 2t\}$. We begin by proving that $M$ is a perfect matching of $P(t)$. By definition, we have that $2m$ divides $2t$. Since $2m = 2^{\depth{t} + 1}$, for any integers $k < 2t$ and $\ell \geq 2t$, there exists $i > \depth{t}$ such that $k(i) \neq \ell(i)$. Additionally, $m(i) = 0$ for any $i > \depth{t}$, so we can define the function $f : \{0, \dots, 2t - 1\} \rightarrow \{0, \dots, 2t - 1\}$ with $f(k) = k \xor m$. For any integers $k$ and $\ell$ with $k \neq \ell$, there exists $i$ such that $k(i) \neq \ell(i)$, but then also $f(k)(i) \neq f(\ell)(i)$ by definition of $f$, so $f$ is injective. From the property that $f(f(k)) = (k \xor m) \xor m = k$ it also follows that $f$ is surjective, hence bijective, and that $f = f^{-1}$. As such, $M$ is a perfect matching of $P(t)$.
    
    Let $p(k) p(\ell) \in M$ with $k < \ell$. Since $m = 2^{\depth{t}}$, we have that $\depth{t}$ is the only natural number for which $k(\depth{t}) \neq \ell(\depth{t})$. From this we first deduce that $\weight(p(k) p(\ell)) = \eps^{\delta(t)}$, and therefore $\weight(M) = \eps^{\delta(t)} t$.
\end{proof}

In addition, $\Malg(t)$ must contain at least $2^{\delta(t)}$ edges of depth $\depth{t}$, and no edges at a lower depth. This builds on the following crucial observation, also illustrated in \autoref{fig:lb:partition}.
\begin{lemma}\label{lem:lb:partition}
    At any time $t \in [n_0]$, the point set $P(t)$ can be partitioned into $m \defeq 2^{\depth{t}+1}$ sets $P_0,\dots,P_{m-1}$ such that: 
    \begin{enumerate}[label=(\roman*)]
        \item $|P_i|$ is odd for all $i \in \{0,\dots,m-1\}$, \label{lem:en:odd}
        \item $\depth{e} \leq \depth{t}$ for all edges $e$ with endpoints in two different sets of the partition, and \label{lem:en:depth}
        \item $\depth{e}>\depth{t}$ for all edges $e$ with endpoints in the same set of the partition.\label{lem:en:depth:two}
    \end{enumerate}
\end{lemma}
\begin{figure}
    \centering
    \includegraphics[width=1\linewidth]{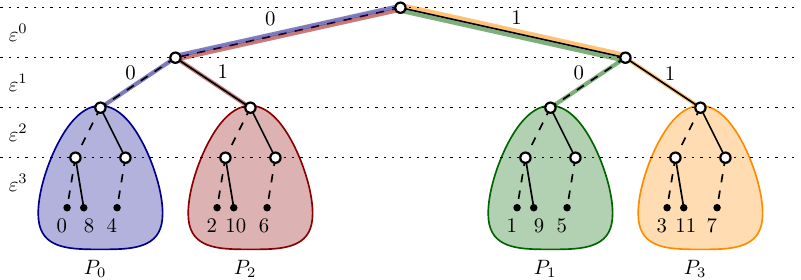}
    \caption{Illustration of \autoref{lem:lb:partition}. The figure depicts the partition of $P(6)$ into $m = 2^{\depth{6}+1} = 4$ subsets, each containing an odd number of points. Edges connecting two parts of the partition have depth at most $1$.} 
    \label{fig:lb:partition}
\end{figure}
\begin{proof}
    We partition the set $P(t)$ into subsets $P_0, \dots, P_{m - 1}$, where we define $$P_j \defeq \{(p(k) \mid k \in \{0,\dots,2t-1\} \text{ and } k \equiv j \pmod{m}\}.$$

    We start by proving statement \ref{lem:en:odd}. Note that $\frac{2t}{m}$ is an odd integer since $m = 2^{\delta(t) + 1} = 2^{\delta(2t)}$ is the largest power of two that divides $2t$. Moreover, for any $j \in \{0, \dots, m - 1\}$ it holds that $|P_j| = \left|\{qm + j \mid q \in \{0, \dots, \frac{2t}{m} - 1\}\}\right| = \frac{2t}{m}$, which is thus odd. 

    We continue with proving statement \ref{lem:en:depth}. Let $e := p(k)p(\ell)$ with $k,\ell <2t$ such that $p(k)$ and $p(\ell)$ are not included in the same set of the partition $P_0,\dots,P_{m-1}$. Thus, $k \not\equiv \ell \pmod{m}$ and by $m=2^{\depth{t}+1}$ there exists $i \leq \depth{t}$ such that $k(i) \neq \ell(i)$. Thus, $\depth{e} = \lsbp{k \xor \ell} \leq \depth{t}$.

   Finally, we prove statement \ref{lem:en:depth:two}. Let $e := p(k)p(\ell)$ with $k,\ell <2t$ such that $p(k)$ and $p(\ell)$ are included in the same set of the partition $P_0,\dots,P_{m-1}$. Then, $k \equiv \ell \pmod{m}$ and by $m = 2^{\depth{t}+1}$ it holds that $k(i)=\ell(i)$ for all $i \leq \depth{t}$ and thus $\depth{e} = \lsbp{k \xor \ell} > \depth{t}$.
\end{proof}
For any depth $d \in \{0, \dots, \log_2 n_0\}$ and $t \in \{0,\dots,n_0\}$, let $\Malg^{d}(t) \defeq \Malg(t) \cap E^d$
be the set of edges of $\Malg(t)$ at depth~$d$.
We can now provide a lower bound for $|\Malg^{\depth{t}}(t)|$.
\begin{lemma} \label{lem:lb:algbound}
At every time~$t$ we have $|\Malg^{\depth{t}}(t)| \geq 2^{\depth{t}}$ and $|\Malg^d(t)| = 0$ for all $d<\depth{t}$. 
\end{lemma}
\begin{proof}
    By \autoref{lem:lb:partition}(i), any feasible matching~at time $t$ 
    contains at least $m/2 = 2^{\depth{t}}$ edges between different sets~$P_j$, and by
    \autoref{lem:lb:partition}(ii) these edges have depth at most $\depth{t}$.
    Now assume for contradiction that $\Malg(t)$ contains an edge $e$ with $\depth{e} < \depth{t}$. Then
    \begin{align*} 
    \alg(t) & \ \geq \ \weight(e) \\
            & \ > \ \tfrac{1}{2}\eps^{\depth{t}-1} \tag*{\llap{(by \autoref{lem:lb:edgebound})}} \\
            & \ = \ \tfrac{1}{2} \cdot (2 \alpha^{\hspace{-0.2mm}*} n_0) \cdot \eps^{\depth{t}} \tag*{\llap{(definition of $\eps$)}}  \\
            & \ \geq \ \alpha(t) \cdot \eps^{\depth{t}} t  \tag*{\llap{(since $\alpha^* = \max_{t \in [n_0]}\alpha(t)$ and $t\leq n_0$)}} \\
            & \ \geq \ \alpha(t) \cdot \opt(t), \tag*{\llap{(by \autoref{lem:lb:optbound})}}  
    \end{align*}
    which contradicts that $\Malg(t)$ is an $\alpha(t)$-approximation. Hence, 
    $|\Malg^d(t)| = 0$ for all $d<\depth{t}$ and $|\Malg^{\depth{t}}(t)| \geq 2^{\depth{t}}$. 
\end{proof}
\autoref{lem:lb:algbound} allows us to lower bound the stability of $\alg$ by 
summing over all $|\Malg^{\depth{t}}(t)|$ for all times $t\leq n_0$. We can do so because for 
any two times $t_1$ and $t_2$ with $\depth{t_1}=\depth{t_2}$ there exists a time $t_3 \in (t_1,t_2)$ 
such that $\depth{t_3}>\depth{t_1}$. Thus, at this point in time, $\alg$ must have deleted 
all edges in $\Malg^{\depth{t_1}}(t_1)$ before including the edges in $\Malg^{\depth{t_2}}(t_2)$.
The next lemma makes this precise.
\begin{lemma} \label{lem:lb:final}
    It holds that $\sum_{t = 0}^{n_0 - 1} |\Malg(t) \Delta \Malg(t+1)| \geq n_0 \log_2 n_0$.
\end{lemma}
\begin{proof}
    We will argue that for any $d \in \{0,\dots,\log_2 n_0\}$, $\alg$ must insert or delete 
    at least $n_0$ edges of depth $d$. To this end, let $C_d \defeq \frac{n_0}{2^d}$. 
    Observe that time $t \leq n_0$ is of depth $d$ if and only if $t= j\cdot 2^d$ for 
    some odd $j \in \{0,\dots,C_d\}$. Moreover, when $j$ is even, then $\depth{t}>d$ for $t=j \cdot 2^d$. 
    Therefore, by \autoref{lem:lb:algbound} it holds that $|\Malg^d(j \cdot 2^d) \Delta \Malg^d((j+1) 2^d)| \geq 2^d$ for any $j  \in \{0,\dots,C_d\}$. We get
    \[
    \sum_{t = 0}^{n_0 - 1} |\Malg^d(t) \Delta \Malg^d(t + 1)|
        \geq \sum_{j = 0}^{C_d - 1} \ |\Malg^d(j \cdot 2^d) \Delta \Malg^d((j+1) 2^d)| \geq C_d \cdot 2^d = n_0.
    \]
    Summing over all $d \in \{0,\dots,\log_2 n_0\}$ yields that the total stability of $\alg$ from $t =0$ until $t=n_0$ is lower bounded by $n_0(\log_2 n_0 + 1)$. 
\end{proof}
\autoref{lem:lb:final} concludes our proof of \autoref{thm:stability-lower-bound} since we have shown that any algorithm that is an $\alpha(n)$-approximation for some function $\alpha$ cannot be $o(\log n)$-stable. We remark that our proof is even stronger than what we state in \autoref{thm:stability-lower-bound} since \autoref{lem:lb:final} lower bounds the \emph{amortized} stability of any approximation algorithm.

\subparagraph{\rm \emph{Remark.}}
The lower bound from \autoref{thm:stability-lower-bound} also holds in the bipartite setting. Consider the exact same construction. Let $\bcov(k)$ be the number of $i \in \Nats$ such that $k(i) = 1$, where $k \in \Nats$. Then $p(k)$ is a server if $\bcov(k)$ is even, otherwise it is a client. Note that $\bcov(k + 1) = \bcov(k) + 1$ if $k$ is even, so exactly one server and one client arrives in each event. Moreover, \autoref{lem:lb:optbound} still holds using the same matching $M$ from the proof. For any edge $p(k) p(\ell) \in M$ with $k < \ell$, we have $\bcov(\ell) = \bcov(k) + 1$, hence $p(k)$ is a server and $p(\ell)$ is a client or vice-versa, thus $M$ is a perfect bipartite matching. The rest of the proof follows identically.

\section{Concluding remarks}
We studied the stability of online algorithms for the minimum-weight perfect matching
problem in~$\Reals^1$. We presented an $O(\sqrt{n})$-stable algorithm with approximation 
ratio~2, which we showed to be optimal among any algorithm with sublinear stability.
We conjecture that $\Omega(\sqrt{n})$ is necessary to achieve approximation ratio~2. 
We also proved that any algorithm with $o(\log n)$ stability must have an unbounded
approximation ratio. It would be interesting to see what approximation ratio
can be achieved with an $f(n)$-stable algorithm for 
$f(n)$ in the range from $\Theta(\log n)$ to $\Theta(\sqrt{n})$.
We conjecture that it is possible, for any integer constant $\alpha>1$, to deterministically maintain 
a $(2\alpha-1)$-approximation with an $O(n^{1/\alpha})$-stable algorithm similar to what 
Goranci~et~al.~\cite{DBLP:conf/icml/GoranciKPSSZ25} achieve using randomization in the bipartite setting.

\bibliography{references}

\end{document}